\documentclass{article}
\usepackage[utf8]{inputenc} 
\usepackage[T1]{fontenc}
\usepackage{url}              
\usepackage{cite}             

\usepackage[cmex10]{amsmath}  
\usepackage{mleftright}       
\mleftright                   

\usepackage{graphicx}         
\usepackage{booktabs}         
\usepackage{mathdots}

\usepackage{amssymb,amsfonts}
\usepackage{enumerate}
\usepackage{comment}
\usepackage{tabularx}
\usepackage{floatrow} 

\newtheorem{theorem}{Theorem}

\newtheorem{lemma}{Lemma}

\newtheorem{remark}{Remark}

\newtheorem{proof}{Proof}

\title{A new construction of an MDS convolutional code of rate $1/2$}
\author{Zita Abreu, Raquel Pinto, Rita Sim\~oes}
\date{}

\begin{document}

\maketitle
\begin{abstract}
Maximum distance separable convolutional codes are characterized by the property that the free distance reaches the generalized Singleton bound, which makes them optimal for error correction. However, the existing constructions of such codes are available over fields of large size. In this paper, we present the unique construction of MDS convolutional codes of rate $1/2$ and degree $5$ over the field $\mathbb{F}_{11}$.
\end{abstract}



\section{Introduction}
Nowadays, all communication systems that work with digitally represented data require the use of error correction codes because all real channels are noisy. One type of error-correcting codes is the convolutional codes. The class of these classical codes is extensively investigated in the literature \cite{bb,cc}. One of the main objectives at the moment is to build codes of a certain rate and degree having as large distance as possible. The distance of a convolutional code measures the robustness of the code since it provides a means to assess its capability to protect data from errors. Codes with longer distance are better because they allow to correct more errors. One type of distance for convolutional codes is the free distance, which is considered for decoding (the process of error correction) when the codeword is fully received. Convolutional codes with maximal free distance (with a certain rate and degree) are called Maximum Distance Separable Codes (MDS). These codes are the ones that present the best performance in error correction among all convolutional codes with fixed rate. 

Up to now, there are not many known constructions of MDS convolutional codes. The first construction was obtained by Justesen in \cite{h} for codes of rate $1/n$ and restricted degrees. In \cite{i} Smarandache and Rosenthal presented constructions of convolutional codes of rate $1/n$  and arbitrary degree $\delta$. However, these constructions require a larger field size than the constructions obtained in \cite{h}. Later, Gluesing-Luerssen and Langfeld presented in \cite{l} a novel construction of convolutional codes of rate $1/n$ with the same field size as the ones obtained in \cite{h} but also with a restriction on the degree of the code. After that, Gluesing-Luerssen, Smarandache and Rosenthal \cite{c} constructed MDS convolutional codes for arbitrary parameters. Lieb and Pinto \cite{d} defined a new construction of convolutional codes of any degree and sufficiently low rate using superregular matrices with a specific property. 

In code constructions, the size of the field is very important for practical implementations since it is directly connected with the computational efficiency of the encoding and decoding algorithms and the complexity of the decoding algorithm, which grows as the size of the field does. In this paper, we present the unique construction of an MDS $(2,1,5)$ convolutional codes over the field $\mathbb{F}_{11}$. The interest of this construction lies in the fact that up to date there is no  constructions for convolutional codes with the same parameters over any field smaller than $\mathbb F_{11}$. 



\section{Preliminaries}
A \textbf{convolutional code} $\mathcal{C}$ of rate $k/n$ is an $\mathbb{F}_{q}[D]$-submodule of $\mathbb{F}_{q}[D]^{n}$ of rank $k$, where $\mathbb{F}_{q}[D]$ is the ring of polynomials with coefficients in the field $\mathbb{F}_{q}$. A $k \times n$ matrix $G(D)$ with entries in $\mathbb{F}_{q}[D]$ whose rows constitute a basis of $\mathcal{C}$ is called a \textbf{generator matrix} for $\mathcal{C}$. This matrix is a full row rank matrix such that \begin{eqnarray}
\mathcal{C} &=&\ \text{Im}_{\mathbb{F}_{q}[D]} \ G(D)\nonumber\\ & = & \{v(D) \in \mathbb{F}_{q}[D]^{n}: v(D) = u(D)G(D) \text{ with } u(D) \in \mathbb{F}_{q}[D]^{k}\}.\nonumber
\end{eqnarray}

Two generator matrices $G_{1}(D),\ G_{2}(D) \in \mathbb{F}_{q}[D]^{k \times n}$ are said to be \textbf{equivalent generator matrices} if $\text{Im}_{\mathbb{F}_{q}[D]} \ G_{1}(D)=\text{Im}_{\mathbb{F}_{q}[D]} \ G_{2}(D),$ which happens if and only if $G_{1}(D) = U(D)G_{2}(D)$ for some unimodular matrix (square polynomial matrix with determinant in $\mathbb{F}_{q}\setminus\{ 0\})$ $U(D) \in \mathbb{F}_{q}[D]^{k \times k}$.

Since two equivalent generator matrices differ by left multiplication with a unimodular matrix, they have equal $k \times k$ full-size minors, up to the multiplication by a nonzero constant. 
The maximum degree of the full-size minors of a generator matrix 
of $\mathcal{C}$ is called the \textbf{degree} of $\mathcal{C}$, and it is normally denoted by $\delta$. Additionally, a convolutional code of rate $k/n$ and degree $\delta$ is also denoted by an $(n, k, \delta)$ convolutional code. 

A matrix $G(D) \in \mathbb{F}[D]^{k \times n}$ is said to be \textbf{left prime} if $G(D) = X(D) \tilde{G}(D)$ for some $X(D) \in \mathbb{F}[D]^{k \times k}$ and $\tilde{G}(D) \in \mathbb{F}[D]^{k \times n}$, then $X(D)$ is unimodular.

Since two equivalent generator matrices differ by left multiplication by a unimodular matrix, if a convolutional code admits a left prime generator matrix then all its generator matrices are left prime and the code is said to be \textbf{noncatastrophic}. A convolutional code that does not admit a left prime generator matrix is said to be \textbf{catastrophic}.

The \textbf{free distance} of a convolutional code measures its capability of detecting and correcting errors introduced during information transmission through a noisy channel and it is defined as $$d_{free}(\mathcal{C})= \min \{wt(v(D)) | v(D) \in \mathcal{C}, v(D) \neq 0\},$$ where $wt(v(D))$ is the Hamming weight of 
$v(D) = \sum_{t=0}^{\deg (v(D))}v_{t}D^{t} \in \mathbb{F}_{q}^{n}[D]$ that is defined as $wt(v(D)) = \sum^{\deg(v(D))}_{t=0} wt (v_t),$ where the weight $wt(v)$ of $v \in \mathbb{F}_{q}^{n}$ is the number of nonzero components of $v$.  

Once channel transmission is complete, $\mathcal{C}$ can detect up to $s$ errors in any received word $w(D)$ if $d_{free}(\mathcal{C}) \geq s + 1$ and can correct up to $t$ errors in $w(D)$ if $d_{free}(\mathcal{C}) \geq 2t + 1$.

In \cite{g} Smarandache and Rosenthal obtained an upper bound for the free distance of an $(n,k,\delta)$ convolutional code $\mathcal{C}$ 
given by $$d_{free}(\mathcal{C}) \leq (n-k)\Big( \Big\lfloor \frac{\delta}{k} \Big\rfloor +1 \Big) + \delta +1.$$

This bound is called the \textbf{generalized Singleton bound}. An $(n,k,\delta)$ convolutional
code 
with free distance equal to the generalized
Singleton bound is called \textbf{Maximum Distance Separable (MDS)} convolutional code. Note that an MDS $(n,1,\delta)$ convolutional code has free distance equal to $n(\delta +1)$.

\section{A construction of MDS Convolutional Codes}
\label{sec:3}

In this section will give a construction of an MDS $(2,1,5)$  convolutional code over $\mathbb{F}_{11}$. First, we will state the following trivial result which will be recurrently used in the proof of the next theorem.

\begin{lemma}
\label{lema}
Let $ G_0 = \begin{bmatrix}
8 & 8
\end{bmatrix}\in \mathbb{F}_{11}^{1 \times 2}$, $ G_1 = \begin{bmatrix}
5 & 6
\end{bmatrix} \in \mathbb{F}_{11}^{1 \times 2}$ and $ G_2 = \begin{bmatrix}
1 & 1
\end{bmatrix} \in \mathbb{F}_{11}^{1 \times 2}$.  and $a,b,c \in \mathbb{F}_{11}.$ Then
\begin{enumerate}
\item if $a$ or $b$ are different from zero, then $wt(aG_{0}+bG_{1}) \geq 1$ and $wt(aG_{1}+bG_{2}) \geq 1$.
\item if $b$ is different from zero, then $wt(aG_{0}+bG_{1}+cG_{2}) \geq 1$.
\end{enumerate}
\end{lemma}

\begin{proof}
It immediately follows from the fact that $G_0$ is a multiple of $G_2$ and that $G_1$ and $G_2$ are linearly independent.
\end{proof}

In the next theorem, we present the first construction up to now of an MDS $(2,1,5)$ convolutional code over the field $\mathbb{F}_{11}$. This is the first construction of an MDS $(2,1,5)$ convolutional code in the literature with a relatively high degree over a small field, see \cite{c},\cite{d},\cite{h} and \cite{l}.

\begin{theorem}
\label{teorema3}
Let $${G}(D) = G_0 + G_1D+ G_2D^2+ G_2D^3+G_1D^4+G_0D^5,$$ with $ G_0 = \begin{bmatrix}
8 & 8
\end{bmatrix}\in \mathbb{F}_{11}^{1 \times 2}$, $ G_1 = \begin{bmatrix}
5 & 6
\end{bmatrix} \in \mathbb{F}_{11}^{1 \times 2}$ and $ G_2 = \begin{bmatrix}
1 & 1
\end{bmatrix} \in \mathbb{F}_{11}^{1 \times 2}$. The $(2, 1, 5)$ convolutional code $\mathcal{C} =\text{Im}_{\mathbb{F}_{11}[D]} \ {G}(D)$ is MDS.
\end{theorem}

\begin{proof}
To prove that $\mathcal{C}$ is MDS we have to show that $d_{free}(\mathcal{C}) = 
12$. Note that $v(D) = u_{0}{G}(D)$ has weight $12$ for every $u_0 \in \mathbb{F}_{11}\setminus\{0\}$. We will show next that $wt(v(D)) \geq 12$ for every $v(D) = \sum_{i \in \mathbb{N}_{0}}v_{i}D^{i} =u(D){G}(D)$ with $u(D)=\sum_{i \in \mathbb{N}_{0}}u_{i}D^{i} \in \mathbb{F}_{11}[D]\setminus\{0\}$ of degree greater or equal than $1$. We will assume without loss of generality that $u_0\neq 0$ and we will consider several cases depending on the degree of $u(D)$.\\

\textbf{Case 1:} If $u(D) = u_0 + u_1D$, with $u_0,u_1 \neq 0$, then $v_0=u_0 G_0$, $v_1=u_0G_1 + u_1G_0 $, $v_2 = u_0G_2 + u_1G_1$, $v_3= (u_0+ u_1)G_2 $, $v_4= u_0G_1+u_1G_2 $, $v_5= u_0G_0+u_1G_1 $ and $ v_6 = u_1G_0$. It is clear that $wt\big(v_i) = 2$ for $i=0,6$ and, by Lemma \ref{lema}, $wt(v_i) \geq 1$, when $i = 1,2,4,5$. It is now necessary to study the weight of $v_3 = \begin{bmatrix}
u_0 + u_1 \ & \ u_0 + u_1 
\end{bmatrix}.$

\textbf{Case 1.1:} If $u_0+u_1 \neq 0$, then $wt(v_3) = 2$. More, $v_1 = \begin{bmatrix}
5u_0+8u_1 \ \ & \ \ 6u_0+8u_1
\end{bmatrix}$ and $v_2=\begin{bmatrix}
u_0+5u_1 \ & \ u_0+6u_1
\end{bmatrix}$.

\textbf{Case 1.1.1:} If $u_0 \neq 5u_1$ and $u_0 \neq 6u_1$ then $wt(v_i)=2$, $i = 1,2$ and therefore  $wt\big(v(D)\big) \geq 12$.

\textbf{Case 1.1.2:} If $u_0 = 5u_1$, then $wt\big(v(D)\big) = 12$ since $v_1= 
u_1\begin{bmatrix}0 & 5\end{bmatrix}$, $v_2=
u_1\begin{bmatrix}10 & 0\end{bmatrix}$, $v_3= 
u_1\begin{bmatrix}6 & 6\end{bmatrix}$, $v_4 = 
u_1\begin{bmatrix}4 & 9\end{bmatrix}$ and $v_5 =
u_1\begin{bmatrix}1 & 2\end{bmatrix}$. If $u_0 = 6u_1$, we also have $wt\big(v(D)\big) = 12$ since 
$v_1 = 
u_1\begin{bmatrix}5 & 0\end{bmatrix}$, $v_2=
u_1\begin{bmatrix}0 & 1 \end{bmatrix}$, $v_3= 
u_1\begin{bmatrix}7 & 7\end{bmatrix}$, $v_4=
u_1\begin{bmatrix}9 & 4\end{bmatrix}$ and $v_5 = 
u_1\begin{bmatrix}9 & 10\end{bmatrix}$. 

\textbf{Case 1.2: } If $u_0+u_1 = 0$ (i.e. $u_1 = 10u_0$) then $v_1 = 
u_0\begin{bmatrix}
8 & 9
\end{bmatrix}$, $ v_2 = 
u_0\begin{bmatrix}
7 & 6
\end{bmatrix}$, $v_3  = u_0\begin{bmatrix}
0 & 0
\end{bmatrix}$, $ v_4=
u_0\begin{bmatrix}
4 & 5
\end{bmatrix}$ and $v_5 = 
u_0\begin{bmatrix}
3 & 2
\end{bmatrix}$. 
Therefore $wt\big(v(D)\big) = 12$.\\

\textbf{Case 2:} If $u(D)=u_0 + u_1D + u_2D^2$, with  $u_0,u_2 \neq 0$ then: $v_0=u_0 G_0$, $v_1=u_0G_1 + u_1G_0$, $v_2=u_0G_2 + u_1G_1 + u_2G_0$, $ v_3=(u_0 + u_1)G_2 + u_2G_1$, $v_4=u_0G_1 + (u_1+u_2)G_2$, $v_5 = u_0G_0 + u_1G_1 + u_2G_2$, $v_6 = u_1G_0+u_2G_1$ and $v_7 = u_2G_0$. Note that $wt(v_i)=2$, for $i =0,7$.

\textbf{Case 2.1:} If $u_1 = 0$ and since $G_0=8G_2$ then $ v_1 = u_0G_1$, $v_2= (u_0 + 8u_2)G_2$, $v_3 = u_0G_2 + u_2G_1$, $v_4= u_0G_1 + u_2G_2$, $v_5 = (8u_0 + u_2)G_2$, and $v_6 = u_2G_1$. It is clear that $wt(v_i)=2$ for $i=1,6$ and, by Lemma \ref{lema}, $wt(v_i) \geq 1$ when $i=3,4$.

If $u_0+8u_2 \neq 0$ (i.e. $u_0 \neq 3 u_2)$ 
then $wt (v_2) = 2$ 
and then $wt\big(v(D)\big) \geq 12.$ If $u_0 = 3u_2$ 
it is easy to see that $wt\big(v(D)\big) \geq 12$.

\textbf{Case 2.2:} If $u_1\neq 0$, by Lemma \ref{lema}, we have that $wt(v_i) \geq 1$ for $i=1,2,3,4,5,6$. Note that $v_1= \begin{bmatrix}5u_0+8u_1 \ & \ 6u_0+8u_1\end{bmatrix} $ and $v_6= \begin{bmatrix}8u_1+5u_2 \ & \ 8u_1+6u_2\end{bmatrix}.$ Thus if $5u_0+8u_1 \neq 0$ (i.e, $u_1 \neq 9u_0$), $6u_0+8u_1 \neq 0$ (i.e, $u_1 \neq 2u_0$), $8u_1+5u_2 \neq 0$ (i.e, $u_2 \neq 5u_1$) and $8u_1+6u_2 \neq 0$ (i.e, $u_2 \neq 6u_1$) then $wt(v_i)=2$ for $i=2,6$ and therefore $wt\big(v(D)\big)\geq 12.$ 

If $u_1=9u_0$ then $v_2= \begin{bmatrix}2u_0+8u_2 & 8u_2\end{bmatrix}$ and $v_6=\begin{bmatrix}6u_0+5u_2 & 6u_0+6u_2\end{bmatrix}$. Therefore, if $2u_0+8u_2 \neq 0$ (i.e, $u_0\neq 7u_2)$, $6u_0+5u_2 \neq 0$ (i.e, $u_0 \neq u_2$) and $6u_0+6u_2\neq 0$ (i.e. $u_0 \neq 10 u_2)$ then $wt(v_2)=wt(v_6)=2$ and consequently $wt\big(v(D)\big) \geq 12.$

If $u_0=7u_2$ then 
$v_3 =u_2\begin{bmatrix}9 &10\end{bmatrix}$, 
$v_5 = u_2\begin{bmatrix}9 &6\end{bmatrix}$ and $v_6 = u_2\begin{bmatrix}3 &4\end{bmatrix}$. Therefore $wt\big(v(D)\big) \geq 12.$


If $u_0=u_2$ then 
$v_2 =u_2 \begin{bmatrix}10 &8\end{bmatrix}$, $v_3 = u_2\begin{bmatrix}4 &5\end{bmatrix}$, $v_4=u_2\begin{bmatrix}4 &5\end{bmatrix}$ and $v_5=u_2\begin{bmatrix}10 &8\end{bmatrix}$. Therefore $wt\big(v(D)\big) \geq 12.$ 


Finally, if $u_0=10u_2$ then 
$v_2 =u_2 \begin{bmatrix}6 &8\end{bmatrix}$, $v_3 =u_2\begin{bmatrix}6 &7\end{bmatrix}$, $v_4 =u_2\begin{bmatrix}9 &8\end{bmatrix}$ and $v_5 =\begin{bmatrix}3 &5\end{bmatrix}$. Therefore $wt\big(v(D)\big) \geq 12.$ 
In the same way $wt\big(v(D)\big) \geq 12$ for $u_1=2u_0$, $u_2=5u_1$ and $u_2 = 6u_1$.\\

\textbf{Case 3:} If $u(D)=u_0 + u_1D + u_2D^2+u_3D^3$, with $u_0,u_3 \neq 0$, then $v_0 = u_0 G_0$, $v_1=  u_0G_1 + u_1G_0 $, $v_2 = u_0G_2 + u_1G_1 + u_2G_0$, $v_3 = (u_0+u_1)G_2 + u_2G_1+u_3G_0$, $v_4 = (u_0+u_3)G_1+(u_1+u_2)G_2$, $v_5 = u_0G_0 +u_1G_1+(u_2+u_3)G_2$, $v_6 = u_1G_0+u_2G_1+u_3G_2 $, $v_7 = u_2G_0 + u_3G_1$ and $v_8 = u_3G_0$. Clearly $wt(v_i)=2$ for $i=0,8$.

\textbf{Case 3.1:} If $u_1 = 0$ and $u_2 = 0$ then:  $v_1=  u_0G_1 $, $v_2 = u_0G_2 $, $v_3 = u_0G_2 +u_3G_0$, $v_4 = (u_0+u_3)G_1$, $v_5 = u_0G_0 +u_3G_2$, $v_6 = u_3G_2 $, and $v_7 = u_3G_1$. Since $wt(v_i)=2$, for $i=1,2,6,7$ it follows that $wt\big(v(D)\big) \geq 12$.

\textbf{Case 3.2:} If $u_1 = 0$ and $u_2 \neq 0$ then $v_1=  u_0G_1 $, $v_2 = u_0G_2 + u_2G_0 $, $v_3 = u_0G_2 + u_2G_1+u_3G_0$, $v_4 = (u_0+u_3)G_1+u_2G_2$, $v_5 =  u_0G_0 +(u_2+u_3)G_2 $, $v_6 = u_2G_1+u_3G_2$ and $v_7 = u_2G_0 + u_3G_1$. Clearly $wt(v_1)=2$ and $wt(v_i)\geq 1$, for $i = 3,4,6,7$. Since $G_0 =8G_2$ we have that $v_2 = (u_0+8u_2)G_2$. Thus, if $u_0+8u_2 \neq 0$, $wt(v_2)=2$ and therefore $wt\big(v(D)\big) \geq 12.$

If $u_0+8u_2 = 0$, i.e., $u_0 =3u_2$, we have that $v_3 = (3u_2+8u_3)G_2+u_{2}G_{1}$, $v_4 = (3u_2+u_3)G_1+u_2G_2$ and $v_5 =(3u_{2}+u_{3})G_2$. If $3u_2+u_3 \neq 0$ (i.e., $u_3 \neq 8u_2$) then $wt(v_5) = 2$ and it follows that $wt\big(v(D)\big) \geq 12$. If $u_3 = 8u_2$ then $v_3 = u_2\begin{bmatrix} 6 & 7\end{bmatrix}$,
$v_6=u_2 \begin{bmatrix} 2 & 3\end{bmatrix}$ and $v_7=u_2\begin{bmatrix} 4 & 1\end{bmatrix}$. Therefore $wt\big(v(D)\big) \geq 12.$

\textbf{Case 3.3:} Using the same reasoning as in Case 3.2, it follows that $wt\big(v(D)\big) \geq 12$ if $u_1 \neq 0$ and $u_2 = 0$.

\textbf{Case 3.4:} If $u_1 \neq 0$ and $u_2 \neq 0$ then by Lemma \ref{lema}, $wt(v_i) \geq 1$, for $i=1,2,3,5,6,7$.

\textbf{Case 3.4.1:} Since $v_1 = \begin{bmatrix}
5u_0+8u_1 \ & \ 6u_0+8u_1
\end{bmatrix}$ and $v_4 = (u_0+u_3)G_1+(u_1+u_2)G_2$, if $5u_0+8u_1 \neq 0$, $6u_0+8u_1 \neq 0$ and $u_0+u_3 \neq 0$ then $wt\big(v(D)\big) \geq 12$.

\textbf{Case 3.4.2:} If $5u_0+8u_1 = 0$, i.e., $u_0=5u_1$, then $v_2 =\begin{bmatrix}
10u_1+8u_2 \ & \ 8u_2
\end{bmatrix} $ and $v_4 = (5u_1+u_3)G_1+(u_1+u_2)G_2$.

\textbf{Case 3.4.2.1:} Therefore if $10u_1+8u_2 \neq 0$ and $5u_1+u_3 \neq 0$ then $wt(v_2)=2$ and $wt(v_4) \geq 1$ and consequently $wt\big(v(D)\big) \geq 12$. 

\textbf{Case 3.4.2.2:} If $10u_1+8u_2 = 0$ (i.e., $u_1 = 8 u_2$) then $v_4 = \begin{bmatrix} 5u_3 \ & \ 7u_2+6u_3 \end{bmatrix}$ has weight $2$ if $7u_2+6u_3 \neq 0$ and therefore $wt\big(v(D)\big) \geq 12$. If $7u_2+6u_3 = 0$ (i.e., $u_2 = 7u_3$) then $v_5 = u_3 \begin{bmatrix} 9  & 10\end{bmatrix}$ and therefore $wt\big(v(D)\big) \geq 12$. 

\textbf{Case 3.4.2.3:} If $5u_1+u_3 = 0$ (i.e., $u_3 =6u_1$) then $v_4 = (u_1+u_2)G_{2}$ and therefore, if $u_1+u_2 \neq 0$, $wt(v_4)=2$ and then $wt\big(v(D)\big) \geq 12$. If $u_1+u_2 = 0$, then $v_2 = u_2 \begin{bmatrix}
9 & 8
\end{bmatrix}$ and $v_6 = u_2 \begin{bmatrix}
2 & 3
\end{bmatrix}$. Therefore $wt\big(v(D)\big) \geq 12$.

\textbf{Case 3.4.3:} In the same way, we prove that $wt\big(v(D)\big) \geq 12$ if $6u_0+8u_1 = 0$ or if $u_0+u_3 = 0$.\\

\textbf{Case 4:} If $u(D)=u_0 + u_1D + u_2D^2+u_3D^3+u_4D^4$, with  $u_0,u_4 \neq 0$ then $v_0 = u_0 G_0 $, $v_1 = u_0G_1 + u_1G_0$, $v_2 = u_0G_2 + u_1G_1 + u_2G_0 $, $v_3 = (u_0 + u_1)G_2 + u_2G_1+u_3G_0$, $v_4 = (u_0 + u_3)G_1+(u_1+u_2)G_2+ u_4G_0$, $v_5 = u_0G_0 +(u_1+u_4)G_1+(u_2+u_3)G_2$, $v_6 = u_0G_0+u_2G_1+(u_3+u_4)G_2$, $v_7 =  u_2G_0 + u_3G_1 +u_4G_2$, $v_8 = u_3G_0 + u_4G_1 $ and $v_9 =u_4G_0$. Clearly $wt(v_i)=2$ for $i=0,9$.

\textbf{Case 4.1:} If $u_1 = 0$, $u_2=0$ and $u_3 = 0$ then $v_1 = u_0G_1$, $v_2 = u_0G_2$, $v_3 = u_0G_2$, $v_7 = u_4G_2$ and $v_8 = u_4G_1$. Since $wt(v_i)=2$, for $i=1,2,3,7,8$ then $wt\big(v(D)\big) \geq 12$.

\textbf{Case 4.2:} If $u_1 = 0$, $u_2=0$ and $u_3 \neq 0$ then $v_1 = u_0G_1$, $v_2 = u_0G_2$, $v_3 = u_0G_2 +u_3G_0$, $v_4 = (u_0 + u_3)G_1+ u_4G_0$, $v_5 =  u_0G_0 +u_4G_1+u_3G_2$, $v_6 = u_0G_0+(u_3+u_4)G_2$, $v_7 = u_3G_1 +u_4G_2$ and $v_8=u_3G_0 + u_4G_1$. We have that $wt(v_i)=2$, for $i=1,2$ and by Lemma \ref{lema}, $wt(v_i) \geq 1$ when $i = 5,7,8$. Since $G_0=8G_2$ then $v_3 = (8u_3+u_0)G_2$ has weight $2$ if $u_0 \neq 3u_3$ and consecutively $w\big(v(D)\big) \geq 12.$ If $u_0=3u_3$ then $v_4 = 4u_3G_1+u_4G_0$ which has weight greater or equal than 1, by Lemma $\ref{lema}$ and therefore $wt\big(v(D)\big) \geq 12$.

\textbf{Case 4.3:} Using a similar reasoning as in Case 4.2, if $u_1 = 0$, $u_2\neq0$ and $u_3 = 0$, we have $wt\big(v(D)\big) \geq 12$.

\textbf{Case 4.4:} In the same way as Case 4.2, $u_1 \neq 0$, $u_2=0$ and $u_3 = 0$, we have $wt\big(v(D)\big) \geq 12$.

\textbf{Case 4.5:} If $u_1 = 0$, $u_2\neq0$ and $u_3 \neq 0$ then $v_1 = u_0G_1$, $v_2 = u_0G_2 + u_2G_0$, $v_3 = u_0G_2 + u_2G_1+u_3G_0$, $v_4 = (u_0 + u_3)G_1+u_2G_2+ u_4G_0$, $v_5 = u_0G_0 +u_4G_1+(u_2+u_3)G_2$, $v_6 = u_0G_0+u_2G_1+(u_3+u_4)G_2$, $v_7 = u_2G_0 + u_3G_1 +u_4G_2$ and $v_8 = u_3G_0 + u_4G_1 $.
Note that $wt(v_1)=2$ and by Lemma \ref{lema}, $wt(v_i) \geq 1$ for $i=3,5,6,7,8$. Since $v_2 = (u_0+8u_2)G_2$ it follows that $wt\big(v(D)\big) \geq 12$ if $u_0+8u_2 \neq 0$. If $u_0+8u_2 = 0$ (i.e. $u_0=3u_2$) then $v_4=(3u_2+u_3)G_1+u_2G_2+u_4G_0$. By Lemma \ref{lema}, if $3u_2+u_3 \neq 0$ (i.e. $u_3 \neq 8u_2$) then $wt(v_4) \geq 1$ and therefore $wt\big(v(D)\big) \geq 12$. If $u_3 = 8u_2$ then $v_3 = u_2(3G_2+G_1)+8u_2G_0 = u_2\begin{bmatrix}
6 & 7
\end{bmatrix}$ and we get $wt\big(v(D)\big) \geq 12$.

\textbf{Case 4.6:} Analogously to the Case 4.5, for $u_1 \neq 0$, $u_2\neq0$ and $u_3 = 0$, we have $wt\big(v(D)\big) \geq 12$.

\textbf{Case 4.7:} If $u_1 \neq 0$, $u_2=0$ and $u_3 \neq 0$ then $v_1 = u_0G_1 + u_1G_0$, $v_2 = u_0G_2 + u_1G_1$, $v_3 = (u_0 + u_1)G_2 +u_3G_0$, $v_4 = (u_0 + u_3)G_1+u_1G_2+ u_4G_0$, $v_5 = u_0G_0 +(u_1+u_4)G_1+u_3G_2$, $v_6 = u_0G_0+(u_3+u_4)G_2$, $v_7 = u_3G_1 +u_4G_2$ and $v_8 = u_3G_0 + u_4G_1$. By Lemma \ref{lema} $wt(v_i) \geq 1$ for $i=1,2,7,8$.

\textbf{Case 4.7.1:} If if $u_0 \neq 5u_1$ and $u_0 \neq 6 u_1$ and since $v_1 = \begin{bmatrix}
5u_0+8u_1 \ & \ 6u_0+8u_1
\end{bmatrix}$ and $v_2 = \begin{bmatrix}
u_0+5u_1 \ & \ u_0+6u_1
\end{bmatrix}$ then $wt(v_i)=2$, for $i=1,2$. Additionally, since $v_7 = \begin{bmatrix}
5u_3+u_4 \ & \ 6u_3+u_4
\end{bmatrix}$ and $v_8 = \begin{bmatrix}
8u_3+5u_4 \ & \ 8u_3+6u_4
\end{bmatrix}$, if $u_4 \neq 6u_3$ and $u_4 \neq 5u_3$, $wt(v_i)=2$, for $i=7,8$ and therefore $wt\big(v(D)\big) \geq 12$.

\textbf{Case 4.7.2:} If $u_0 = 5u_1$, then $v_3 = (6u_1+8u_3)G_2$, $v_4 = (5u_1+u_3)G_1+(u_1+8u_4)G_2$, $v_5 = (7u_1+u_3)G_2+(u_1+u_4)G_1$, $v_6 = (7u_1+u_3+u_4)G_2$, $v_7 = u_3G_1+u_4G_2$ and $v_8 = u_3G_0+u_4G_1$.  

\textbf{Case 4.7.2.1:} If $5u_1+u_3 \neq 0$ (i.e. $u_3 \neq 6u_1$), $6u_1+8u_3 \neq 0$ (i.e. $u_3 \neq 2u_1$) and $u_1+u_4 \neq 0$ (i.e. $u_4 \neq 10u_1$) then $wt(v_3)=2$ and $wt(v_i) \geq 1$, for $i=4,5.$

\textbf{Case 4.7.2.2:} If $u_3 = 6u_1$ then $v_3 = u_1 \begin{bmatrix} 10 & 10 \end{bmatrix}$ and $v_4 = (u_1+8u_4)G_2$.

If $u_1+8u_4 \neq 0$ (i.e. $u_4 \neq 4u_1$) then $wt(v_4) = 2$ and therefore $wt\big(v(D)\big) \geq 12$.

If $u_4 = 4u_1$ then $v_6 = u_1 \begin{bmatrix} 6 & 6 \end{bmatrix}$, $v_5 = 2u_1G_2+5u_1G_1$ and $v_7 = u_1 \begin{bmatrix} 1 & 7 \end{bmatrix}$. Therefore $wt(v_i)=2$ for $i=6,7$, $wt(v_5) \geq 1$ and consequently $wt\big(v(D)\big) \geq 12$.

\textbf{Case 4.7.2.3:} Similar to Case 4.7.2.2. If $u_3 = 2u_1$, then $wt\big(v(D)\big) \geq 12$.

\textbf{Case 4.7.2.4:} Similar to Case 4.7.2.2. If $u_4 = 10u_1$ then $wt\big(v(D)\big) \geq 12$.

\textbf{Case 4.7.3:} Similar to Case 4.7.3. If $u_0 = 6u_1$ then $wt\big(v(D)\big) \geq 12$.

\textbf{Case 4.7.4:} If $u_4 = 6u_3$ then $v_1=u_0G_1+u_1G_0$, $v_2=u_0G_2+u_1G_1$, $v_3 = (u_0+u_1+8u_3)G_2$, $v_4 = (u_0+u_3)G_1+(u_1+4u_3)G_2$, $v_5=(u_1+6u_3)G_1+(8u_0+u_3)G_2$, $v_6 = (8u_0+7u_3)G_2$, $v_7 = 
u_3 \begin{bmatrix} 0 & 1 \end{bmatrix}$ and $v_8 = 
u_3 \begin{bmatrix} 5 & 0 \end{bmatrix}$.

\textbf{Case 4.7.4.1:} If $u_0+u_3 \neq 0$, $8u_0+7u_3 \neq 0$ and $u_1+6u_3 \neq 0$ then $wt(v_i) \geq 1$, for $i=4,5$ and $wt(v_6)=2$ and therefore $wt\big(v(D)\big) \geq 12$.

\textbf{Case 4.7.4.2:} If $u_0+u_3 = 0$, i.e, $u_0 = 10u_3$ then $v_4=(u_1+4u_3)G_2$ and $v_6 =  u_3\begin{bmatrix} 10 & 10 \end{bmatrix}$. So $wt(v_6)=2$ and $wt(v_4)=2$ if $u_1+4u_3 \neq 0$ and therefore $wt\big(v(D)\big) \geq 12$. If $u_1+4u_3 = 0$, i.e. $u_1 = 7u_3$ then $v_5 = 
u_3\begin{bmatrix} 3 & 5 \end{bmatrix}$ and $v_3 = u_3\begin{bmatrix} 3 & 3 \end{bmatrix}$. So $wt(v_i)=2$, for $i=3,5$ and consequently $wt\big(v(D)\big) \geq 12$.

\textbf{Case 4.7.4.3:} Similarly to Case 4.7.4.2, it is possible to conclude that if $8u_0+7u_3 = 0$, i.e., $u_0 = 6u_3$ then $wt\big(v(D)\big) \geq 12$.

\textbf{Case 4.7.4.4:} If $u_1 +6u_3 = 0$, i.e., $u_1=5u_3$ then $v_4 = (u_0+u_3)G_1+9u_3G_2$, $v_5 = (8u_0+u_3)G_2$ and $v_3 = (u_0+2u_3)G_2$. Then $wt(v_i)=2$, for $i=3,5$ if $8u_0+u_3 \neq 0$ and $u_0+2u_3 \neq 0$. If $8u_0+u_3 = 0$ (i.e.  $u_3 = 3u_0$) then $v_4 = u_0 \begin{bmatrix} 5 & 9 \end{bmatrix}$ and $v_3 = u_0\begin{bmatrix} 7 & 7 \end{bmatrix}$. So $wt(v_i)=2$, for $i=3,4$ and therefore $wt\big(v(D)\big) \geq 12$. If $u_0+2u_3=0$ (i.e. $u_3 = 5u_0$) then $v_4 = u_0\begin{bmatrix} 9 & 4 \end{bmatrix}$ and $v_5 = u_0\begin{bmatrix} 2 & 2 \end{bmatrix}$. So $wt(v_i)=2$, for $i=4,5$ and consequently $wt\big(v(D)\big) \geq 12$.

\textbf{Case 4.7.5:} Similar to Case 4.7.4. If $u_4= 5u_3$ then $wt\big(v(D)\big) \geq 12$.

\textbf{Case 4.8:} If $u_1 \neq 0$, $u_2\neq0$ and $u_3 \neq 0$ then, since $G_0 =8G_2$, $v_1=u_0G_1+u_1G_0$, $v_2 = (u_0+8u_2)G_2 + u_1G_1$, $v_3 =(u_0 + u_1+8u_3)G_2 + u_2G_1$, $v_4 = (u_0 + u_3)G_1+(u_1+u_2+8u_4)G_2$, $v_5 = (8u_0+u_2+u_3)G_2+(u_1+u_4)G_1$, $v_6 = u_2G_1+(u_3+u_4+8u_0)G_2$, $v_7 = u_3G_1 +(8u_2+u_4)G_2$ and $v_8 = u_3G_0 + u_4G_1$. By Lemma \ref{lema}, $wt(v_i) \geq 1$ for $i=1,2,3,6,7,8$. Note that $v_1 = \begin{bmatrix} 5u_0+8u_1 \ & \ 6u_0 + 8u_1
\end{bmatrix}$ and $v_8 = \begin{bmatrix} 8u_3+5u_4 \ & \ 8u_3 + 6u_4
\end{bmatrix}$.

\textbf{Case 4.8.1:} If $ 5u_0+8u_1 \neq 0$, $6u_0 + 8u_1 \neq 0$, $8u_3+5u_4 \neq 0$ and $8u_3 + 6u_4 \neq 0$ then $wt(v_i)=2$ for $i=2,8$ and consequently $wt\big(v(D)\big) \geq 12$.

\textbf{Case 4.8.2:} If $ 5u_0+8u_1 = 0$ (i.e. $u_0=5u_1$) then $v_1 = \begin{bmatrix} 0 & 5
\end{bmatrix}$, $v_2 = \begin{bmatrix} 10u_1+8u_2 & 8u_2
\end{bmatrix}$ and $v_8 = \begin{bmatrix} 8u_3+5u_4 & 8u_3+6u_4
\end{bmatrix}$

\textbf{Case 4.8.2.1:} If $10u_1+8u_2 \neq 0$, $8u_3+5u_4 \neq 0$ and $8u_3+6u_4 \neq 0$ then $wt(v_i)=2$ for $i =2,8$ and thus $wt\big(v(D)\big) \geq 12$.

\textbf{Case 4.8.2.2:} If $10u_1+8u_2 = 0$ (i.e. $u_2 = 7u_1$) then $v_5 = (8u_1+u_3)G_2+(u_1+u_4)G_1$ and $v_4 = (5u_1+u_3)G_1+(8u_1+8u_4)G_2$.

\textbf{Case 4.8.2.2.1:} If $u_1+u_4 \neq 0$ (i.e. $u_4 \neq 10u_1$)  and $5u_1+u_3 \neq 0$ (i.e. $u_3 \neq 6u_1$) then $wt(v_i) \geq 1$ for $i=4,5$ and consequently $wt\big(v(D)\big) \geq 12$.

\textbf{Case 4.8.2.2.2:} If $u_4 = 10u_1$ then $v_4 = (5u_1+u_3)G_1 $ and $v_5 =(8u_1+u_3)G_2$.

\textbf{Case 4.8.2.2.2.1:} If $5u_1+u_3 \neq 0$ (i.e. $u_3 \neq 6 u_1$) then $wt(v_4)=2$ and consequently $wt\big(v(D)\big) \geq 12$.

If $u_3 = 6 u_1$ then $v_5 = u_1 \begin{bmatrix} 3 & 3
\end{bmatrix}$. Therefore $wt(v_5)=2$ and $wt\big(v(D)\big) \geq 12$.

\textbf{Case 4.8.2.2.3:} Similar to Case 4.8.2.2.2. If $5u_1+u_3 = 0$ then $wt\big(v(D)\big) \geq 12$.

\textbf{Case 4.8.2.3:} Similar to Case 4.8.2.2. If $8u_3+5u_4 = 0$ then $wt\big(v(D)\big) \geq 12$.

\textbf{Case 4.8.2.4:} Similar to Case 4.8.2.2. If $8u_3+6u_4 = 0$ then $wt\big(v(D)\big) \geq 12$.

\textbf{Case 4.8.3:} Similar to Case 4.8.2.1, if $6u_0+8u_1 = 0$ or $8u_3+5u_4 =0$ or $8u_3+6u_4 = 0$ then $wt\big(v(D)\big) \geq 12$.\\

\textbf{Case 5:} Using the same reasoning as before it can be proved that $wt\big(v(D)\big) \geq 12$ if $u(D)=u_0 + u_1D + u_2D^2+u_3D^3+u_4D^4 +u_5D^5$, with $u_0,u_5 \neq 0$.\\

Finally, let us consider the last case.

\textbf{Case 6:} If $u(D)=u_0 + u_1D + \cdots +  u_nD^n$, with $u_0,u_n \neq 0$ and $n > 5$, then $v_0 = u_0 G_0$, $v_1 = u_0G_1 + u_1G_0$, $v_2 = u_0G_2 + u_1G_1 + u_2G_0$, $v_3 = u_0G_2 + u_1G_2 + u_2G_1+u_3G_0$, $v_4 = u_0G_1+u_1G_2+u_2G_2+u_3G_1 + u_4G_0$, $v_5 = u_0G_0+u_1G_1+u_2G_2+u_3G_2+u_4G_1+u_5G_0$, $\cdots$ ,$v_{n+1}=u_{n}G_1 + u_{n-1}G_2 + u_{n-2}G_2 + u_{n-3}G_1 + u_{n-4}G_{0}$, $v_{n+2}=u_{n}G_2 + u_{n-1}G_2 + u_{n-2}G_1 + u_{n-3}G_0$, $v_{n+3}=u_{n}G_2 + u_{n-1}G_1 + u_{n-2} G_0$, $v_{n+4}=u_nG_1 + u_{n-1}G_0$ and $v_{n+5}= u_nG_0$. As before, $wt(v_0)=wt(v_{n+5})=2$.


Let us first show that the weight of $v(D)\vert_{[0,5]}:=v_0+v_1D+v_2D^2+v_3D^3+v_4D^4+v_5D^5$ has weight greater or equal  than $6$ for all $u_0, u_1,u_2,u_3,u_4,u_5  \in \mathbb{F}_{11}$ with $u_0 \neq 0.$

\textbf{Case 6.1:} Let us consider first the case $u_1=0$. Then $v_1=u_0G_1$ has weight 2, $v_2=(u_0+8u_2)G_2$, $v_3=(u_0 +8u_3)G_2+u_2G_1$ and $v_4=(u_2+8u_4)G_2+(u_0+u_3)G_1$. In this case, if $u_2 \neq 0$ and $u_0+u_3 \neq 0$, then, by Lemma \ref{lema}, $wt(v_3) \geq 1$ and  $wt(v_4) \geq 1$ and $wt\big(v(D)\vert_{[0,5]}\big) \geq 6$. 

If $u_2 = 0$, then $v_2=u_0G_2$ and therefore $wt(v_2)=2$, which implies that $wt\big(v(D)\vert_{[0,5]}\big) \geq 6$. 

If $u_0 + u_3 = 0$, i.e., $u_0=10u_3$, then $v_2 =( 10u_3+8u_2)G_2$ and $v_3=7u_3G_2+u_2G_1$. In this case if $10u_3+8u_2 \neq 0$ then $wt(v_2) = 2$ and therefore $wt\big(v(D)\vert_{[0,5]}\big) \geq 6$. If $10u_3 + 8u_2 = 0$, i.e., $u_3 = 8u_2$, then $v_3=
u_2 \begin{bmatrix} 6 & 7 \end{bmatrix}$ and therefore $wt(v_3) = 2 $ which implies that $wt\big(v(D)\vert_{[0,5]}\big) \geq 6$.

\textbf{Case 6.2:} Let us consider now that $u_1 \neq 0$. Then $v_1= \begin{bmatrix}
5u_0 + 8u_1 \ & \ 6u_0+8u_1
\end{bmatrix}$.

\textbf{Case 6.2.1:} If $u_0 \neq 6u_1$ and $u_0 \neq 5u_1$ then $wt(v_1)=2$ and $wt(v_2)\geq 1$, by Lemma \ref{lema}.

If $u_2 \neq 0$ then also $wt(v_3) \geq 1$ and therefore $wt\big(v(D)\vert_{[0,5]}\big) \geq 6$. 

If $u_2 = 0$ then $v_3=(u_0+u_1+8u_3)G_2 $ and $v_4=(u_1+8u_4)G_2+(u_0+u_3)G_1$. Since $v_4 = \begin{bmatrix} 0 & 0 \end{bmatrix}$ then $
u_1= 3u_4 \text{ and }
u_0=10u_3$. If $u_1=3u_4$ and $u_0=10u_3$ it follows that $v_3=
(7u_3+3u_4)G_2$ and $v_2 = 10u_3G_2+3u_4G_1$.

If $7u_3+3u_4=0$, i.e. $u_3=9u_4$ then 
$v_2=
u_4\begin{bmatrix} 6 & 9 \end{bmatrix}$ which has weight $2$.

If $7u_3+3u_4 \neq 0$ then $wt(v_3)=2$. Thus $wt\big(v(D)\vert_{[0,5]}\big) \geq 6$. 

\textbf{Case 6.2.2:} Let us consider now that $u_0=6u_1$. Then $v_1=
u_1\begin{bmatrix} 5 & 0 \end{bmatrix}$, $v_2 = 6u_1G_2+u_1G_1+u_2G_0$, $v_3 = 7u_1G_2+u_2G_1+u_3G_0$ and $v_4=(6u_1+u_3)G_1+(u_1+u_2)G_2+u_4G_0$. Then we have that $wt(v_1)=1$ and, by Lemma \ref{lema}, $wt(v_2) \geq 1$.

\textbf{Case 6.2.2.1:} If $u_2 \neq 0$ and $6u_1+u_3 \neq 0$ (i.e. $u_3\neq 5u_1$) then $wt(v_3) \geq 1$ and $wt(v_4)\geq 1$ and therefore $wt\big(v(D)\vert_{[0,5]}\big) \geq 6$.

\textbf{Case 6.2.2.2:} If $u_2 =0$ then  $v_2=
u_1\begin{bmatrix} 0 & 1 \end{bmatrix}$, $v_3=
(7u_1+8u_3)G_2$, $v_4=(6u_1+u_3)G_1+u_1G_2+u_4G_0$.

\textbf{Case 6.2.2.1.1:} If $7u_1+8u_3 \neq 0 $ (i.e. $u_1\neq 2u_3$) then $wt(v_3)=2$ and $wt\big(v(D)\vert_{[0,5]}\big) \geq 6$. 

\textbf{Case 6.2.2.1.2:} If $u_1=2u_3$ then $v_4=
\begin{bmatrix} u_3+8u_4 & 3u_3+8u_4 \end{bmatrix}$. Thus, if $u_3+8u_4 \neq 0$ and $3u_3+8u_4 \neq 0$ (i.e. $u_3 \neq 3u_4$ and $u_3 \neq u_4$) then $wt(v_4)=2$ and $wt\big(v(D)\vert_{[0,5]}\big) \geq 6$. 

If $u_3=3u_4$ then $wt(v_4)=1$ and $v_5=(3u_4+u_5)G_0+7u_4G_1+3u_4G_2$ is such that $wt(v_5) \geq 1$ and $wt\big(v(D)\vert_{[0,5]}\big) \geq 6$. 

If $u_3=u_4$ then $vt(v_4)=1$ and $v_5=(u_4+u_5)G_0+3u_4G_1+u_4G_2$ has also weight greater or equal than $1$. So we conclude that $wt\big(v(D)\vert_{[0,5]}\big) \geq 6$. 

\textbf{Case 6.2.2.3:} On the other hand, if $u_3=5u_1$ then $v_1=u_1\begin{bmatrix} 5 & 0 \end{bmatrix}$, $v_2=6u_1G_2+u_1G_1+u_2G_0$, $v_3=
\begin{bmatrix} 3u_1+5u_2 \ & \ 3u_1+6u_2 \end{bmatrix}$. Thus $wt(v_1)=1$ and $wt(v_2) \geq 1$ by Lemma \ref{lema}.

\textbf{Case 6.2.2.3.1:} If $3u_1+5u_2 \neq 0$ and $3u_1+6u_2 \neq 0$ (i.e. $u_1 \neq 2u_2$ and $u_1 \neq 9u_2$) then $wt(v_3)=2$ and therefore $wt\big(v(D)\vert_{[0,5]}\big) \geq 6$. 

\textbf{Case 6.2.2.3.2:} If $u_1=2u_2$ then $wt(v_3)=1$ and $v_4=(7u_1+8u_4)G_2$.

\textbf{Case 6.2.2.3.2.1:} If $7u_1+8u_4 \neq 0$, then $wt(v_4)=2$ and therefore $wt\big(v(D)\vert_{[0,5]}\big) \geq 6$. 

\textbf{Case 6.2.2.3.2.2:} If $7u_1+8u_4=0$, i.e., $u_1=2u_4$ and then $v_2=
u_1\begin{bmatrix} 4 & 5 \end{bmatrix}$ which has weight $2$ and then $wt\big(v(D)\vert_{[0,5]}\big) \geq 6$. 

\textbf{Case 6.2.2.3.3:} Similar to Case 6.2.2.3.2. If $u_1 = 9u_2$ then $wt\big(v(D)\vert_{[0,5]}\big) \geq 6$. 

\textbf{Case 6.2.3:} In the same way as the Case 6.2.2, it is possible to conclude that if $u_0=5u_1$ then $wt\big(v(D)\vert_{[0,5]}\big) \geq 6$.

Thus, we have proven that $wt\big(v(D)\vert_{[0,5]}\big) \geq 6$. Note that $v_n = \hat{v_5},v_{n+1}=\hat{v_4}, v_{n+2} = \hat{v_3}, v_{n+3} = \hat{v_2}, v_{n+4} = \hat{v_1} v_{n+5}= \hat{v_0}$ for $$\hat{v}(D) = (u_n + u_{n-1}D + u_{n-2}D^2 + u_{n-3}D^3 + u_{n-4}D^4 + u_{n-5}D^5){G}(D).$$

Since $wt(\hat{v_0}+\hat{v_1}D+\hat{v_2}D^2+\hat{v_3}D^3+\hat{v_4}D^4+\hat{v_5}D^5) \geq 6$ then $$wt(v_n D^{n}+ v_{n+1}D^{n+1} + v_{n+2}D^{n+2} + v_{n+3}D^{n+3} + v_{n+4}D^{n+4} + v_{n+5}D^{n+5}) \geq 6$$ for all $u_{n-5}, u_{n-4}, u_{n-3}, u_{n-2}, u_{n-1}, u_n \in \mathbb{F}_{11}$ with $u_n \neq 0$ thus $wt\big(v(D)\big) \geq 12$.
\end{proof}




Note that the convolutional code defined in Theorem \ref{teorema3} is catastrophic, since $$G(D)=(1+D)[(1+10D+D^2+10D^3+D^4)G_{0}+(D+10D^{2}+D^{3})G_{1}+D^{2}G_{2}],$$ because $1+D$ is not a unimodular factor.

A noncatastrophic convolutional code can be constructed using a similar construction, by slightly changing the last three coefficients of $G(D)$. For example, $$G(D)=G_0+G_1D+G_2D^2+G_2D^3+aG_1D^4+bG_0D^5$$ with $ G_0 = \begin{bmatrix}
8 & 8
\end{bmatrix}\in \mathbb{F}_{11}^{1 \times 2}$, $ G_1 = \begin{bmatrix}
5 & 6
\end{bmatrix} \in \mathbb{F}_{11}^{1 \times 2}$ and $ G_2 = \begin{bmatrix}
1 & 1
\end{bmatrix} \in \mathbb{F}_{11}^{1 \times 2}$ is a generator matrix of a noncatastrophic $(2,1,5)$ convolutional code for every $(a,b) \in \mathbb{F}_{11}^{2}\setminus \{(1,1)\}$. Is still an open problem to find it this code is MDS for some $(a,b) \in \mathbb{F}_{11}^{2}\setminus \{(1,1)\}$.

\begin{remark}
    The coefficients of the generator matrix $G(D)$ defined in the Theorem \ref{teorema3}, $(G_0, G_1, G_2, G_2, G_1, G_0)$ are in a palindrome format and are such that the generator matrix defined by the first three coefficients \begin{equation} \label{+}
\tilde{G}(D) = G_0+G_1D+G_2D^2
\end{equation} is an MDS $(2,1,2)$ convolutional code over $\mathbb{F}_{11}$ defined by Justesen in \cite{h} (see also \cite{cc}). Justesen gave (the first) construction of MDS convolutional codes of rate $1/n$. In particular, for $n=2$, the proposed construction is given by the following theorem.

\begin{theorem}[\cite{h}]
\label{teorema}
For $n =2$ and $|\mathbb{F}_{q}| \geq 3$, set $s_{2}:=\Big\lceil \frac{|\mathbb{F}_{q}|-1}{2} \Big\rceil$ and $\delta := \Big\lfloor \frac{2}{9}| \mathbb{F}_{q}| \Big\rfloor$. Moreover, let $\alpha$ be a primitive element of $\mathbb{F}_{q}$, and set $g_{1}(x) := (x-\alpha)(x-\alpha^{2})$ and $g_{2}(x) := g_{1}(x\alpha^{-s_{2}})$. Then $G(D) = [g_1(D) \ g_{2}(D)]$ is the generator matrix of an MDS $(2, 1, 2)$ convolutional code.
\end{theorem}

Considering $\delta = 2$, the field $\mathbb{F}_{11}$ and $\alpha = 2$ as a primitive element of $\mathbb{F}_{11}$, the Theorem \ref{teorema} gives the generator matrix $\tilde{G}(D)$ considered in (\ref{+}) of an MDS convolutional code, and repeating the coefficients of $\tilde{G}(D)$ in reverse order we define the generator matrix of Theorem \ref{teorema3}. However, this reasonig does not apply for all the codes defined by Theorem \ref{teorema}. In particular Theorem \ref{teorema} gives an MDS $(2,1,2)$ convolutional code over the fields $\mathbb F_9$ and $\mathbb F_{11}$, by considering any primitive element of the field. However, we checked that when we take $$\tilde{G}(D)=G_0+G_1D+G_2D^2$$ defined in Theorem \ref{teorema} by considering $\delta = 2$ and $\mathbb{F}_{9}$, the generator matrix $$G(D) = G_0 +G_1D+G_2D^2+G_2D^3+G_1D^4+G_0D^5$$ defines a $(2,1,5)$ convolutional code which is not MDS, for any chosen primitive element of the fields. Moreover, the same happens when we consider the field $\mathbb F_{11}$ and a primitive element of $\mathbb F_{11}$ different from $2$.
\end{remark}



\section*{Acknowledgements}
This work is supported by The Center for Research and Development in Mathematics and Applications (CIDMA) through the Portuguese Foundation for Science and Technology (FCT), UIDB/04106/2020 and UIDP/04106/2020. The work of the first author was also supported by FCT grant UI/BD/151186/2021.

\end{document}